\documentclass[conference]{IEEEtran}
\usepackage{float}

\usepackage[margin=.9in]{geometry} 
\usepackage{lipsum}
\usepackage{fancyhdr}
\usepackage{setspace}
\usepackage{amsmath}
\usepackage{algorithm,algcompatible,lipsum}
\usepackage{amsmath, amsthm, amssymb}
\usepackage{graphicx,curves}

\newtheorem{prop}{Theorem}

\title{Task and Energy Aware Node Placement in  Wirelessly Rechargeable WSNs}
	\author{Md Solimul Chowdhury \\mdsolimu@ualberta.ca\\ Department of Computing Science \\ University of Alberta, Edmonton, Canada.}

\begin{document}
	\maketitle
    \date{}
	
	\onehalfspacing
	\pagestyle{plain}
    \begin{abstract}
In this paper, we present a novel problem of optimal placement of sensor nodes in wirelessly rechargeable Wireless Sensor Networks (WSNs) wrt. a charging requirement constraint and a task utility requirement constraint. We call this problem Task and Energy Aware Node Placement (TENP) problem. We have devised an algorithm to solve the TENP problem. Our theoretical analysis shows that the devised algorithm is an incomplete but tractable method for solving TENP. We have performed empirical evaluation of the devised algorithm for three different version of TENP in various experimental settings. The experimental results reveal numbers of interesting insights on the relationship between sensor node placement, charge harvest and task utility.
    \end{abstract}
	\section{Introduction}
	
	A \textit{Wireless Sensor Networks (WSN)} consists of two types of devices- \textit{sensor nodes} and \textit{sink nodes}. In a WSN, sensor nodes are equipped with sensing capabilities, which report the sensed data to the sink nodes for further processing. With these operational features, over the last years, WSNs has become the central technology for application domains, where autonomous monitoring of environment and reporting of events of interest are required.  Patient monitoring, fire detection, wildlife monitoring \cite{applicationWSN} - are only a few examples of WSN applications.  However, energy-efficiency of  WSNs is still a serious concern. In a typical WSN, the sensor nodes are powered by small sized batteries of limited capacity, which can severely limit the lifetime of that WSN. This energy solution is not sufficient for many applications that require power supply for extended period of time. Additionally, replenishing or replacing the batteries of sensor nodes are often a difficult task, as most of these applications are deployed in remote and inaccessible environments. Some alternative solutions to the energy-efficiency issues of WSNs have been proposed in the literature, such as, energy-aware networking protocols \cite{energy-protocol-survey}, ambient energy harvesting (\cite{ambient1,ambient2}), Radio-Frequency based Energy Transfer (RFET) \cite{rfet} technology etc.
	
For the RFET based approach, \textit{Energy Transmitters (ETs)} transmit energy signal as Electromagnetic (EM) waves, which are received by the sensor nodes. One property of EM wave is that it decays over distance. Thus the placement of ETs and sensor nodes are crucial to the life-time and utilization of a WSN.

In a WSN, different groups of sensor nodes are responsible for accomplishing different tasks. The utilization of a sensor node wrt. a task depends on the distance between that sensor node and that task. Depending on the importance of the tasks, their utilization requirements for the sensor nodes can vary. In other words, the more important a task is, the more it needs to utilize its associated group of sensor nodes. Thus the distance between sensor nodes associated to a task, need to respect the utilization requirement set by that task.

In this paper, we study optimal placement of sensor nodes that takes both of the above scenarios into consideration, namely - minimization of 1) distance between ETs and sensor nodes and 2) distance between tasks and their associated sensor nodes, subject to the following two constraints: a) a minimum charging requirement for each of the sensor nodes must be satisfied and b) the sensor nodes must satisfy a minimum utility requirement set up by their associated tasks.  We name this problem as \textit{Task and Energy Aware Node Placement (TENP)}.

The rest of this paper is organized as follows: In the next section, we present a general overview of energy solutions of WSNs. In Section 3, we review some existing works about energy, transmission and task aware node placement in WSNs. We formulate the TENP problem in Section 4. In Section 5, we present the devised algorithms for solving the TENP problem. Theoretical analysis of the devised algorithms are presented in Section 6. We present our empirical evaluation of the devised algorithms in Section 7. In Section 8, we briefly summarize our work and give some future directions.

\section{General Overview of Energy Solutions of WSNs}
The research on energy issues of WSN can be divided into two groups. While, the first group perceives the energy-efficiency problem of WSNs as the problem of designing energy efficient network protocols for WSNs, the second group incorporates a distinct technology, named - Radio Frequency Energy Transfer (RFET). In this section, we briefly present a general overview of these two groups.
 \subsection{Designing Energy-Aware Network Protocols}
Designing energy efficient networking protocols can be categorized into two types: \textit{Clustering Based approach} and \textit{Tree Based approach} \cite{energy-protocol-survey}. In  the clustering based approach, a WSN is divided into  a set of clusters of sensor nodes. The clustering approach makes a WSN more energy efficient by a) preventing transfer of replicated message among the sensor nodes in a given cluster and b) establishing an energy-efficient route within a local cluster. In the tree-based approach, the whole network is organized as a tree, where the sink node is the root node of that tree and the other nodes (intermediate nodes and leaves) are sensor nodes. Given an intermediate node \textit{I}, the idea is to have \textit{I} (residing at a higher level than \textit{S}) aggregate the data sent by the nodes residing below $I$. As a result of aggregation, network traffic can be reduced, which in turn, can increase the energy-efficiency and life-time of a WSN.

Design and implementation of energy efficient routing protocol has increased the energy efficiency of WSNs. However, there are some applications, where sensor nodes need to be deployed without getting replaced or manually recharged. For example, the desired lifetime of sensor nodes for applications, such as environmental/wildlife monitoring, is on the scale of years. This type of application requirements, warrant alternative types of energy solutions for WSNs.

\subsection{RFET based WSNs}
A WSN based on Radio Frequency Energy Transfer (WRFET) is composed of sensor nodes, sink nodes and a special type of nodes called \textit{Energy Transmitter (ET)} nodes. In a WRFET, ET nodes are responsible for wirelessly recharging sensor nodes. The wireless transfer is performed by using the so called \textit{Magnetic Resonance Coupling (MRC)} technique. The MRC technique employs a single copper loop of small radius both at the energy transmitter end and sensor node’s receiving end \cite{notesOnEnergyHarvesting}.

The ETs in a given WRFET can be static or mobile. The locations of static ETs are fixed and the locations of mobile ETs are not fixed, as this type of ETs can change their locations and transmits energy to the sensor nodes. An ET is capable of charging sensor nodes within a fixed radius.
\subsubsection{Ambient Energy Harvesting and WRFET}
Harvesting ambient energy and then using the harvested energy to operate a WSN, is another relatively new development in the landscape of energy-efficiency research for WSN.  We briefly summarize this approach based on our reading of \cite{ambient1,ambient2}.

When energy is harvested from the ambient, the  nodes can have ample supply of energy. In non-RF-based Energy harvesting schemes, energy harvesting can be performed primarily in two ways: i) the ambient energy can be directly converted and used by the sensor nodes, ii) the sensor nodes store the harvested energy in battery storage and use it from those battery storages. In RF-based energy harvesting schemes, dedicated ETs are used, which has the following two capabilities a) harvesting energy from the ambient source and b) transferring the harvested energy wirelessly to the sensor nodes by using RF.  Solar energy, mechanical energy, thermal energy and commercial energy harvester are the main types of energy source of energy harvesting for WSNs.

Despite  being a game-changing solution to the energy issues of WSNs, ambient energy harvesting has a major drawback. Most of the energy harvesting technology are unpredictable and sometime uncontrollable. For example, on a cloudy day, the absence of sun can create severe energy crisis in a WSN that harvests energy from solar energy. This inherent nature of ambient energy, has driven the research community to develop alternative solutions to the energy issue of WSNs.

\subsubsection{Constant Power Source based WRFET}
Constant Power Source based WRFET  uses a constant power source for the ETs, instead of harvesting energy from the ambient. In this type of WRFET, the ETs can get recharged by getting connected to dedicated power outlets. For the static ETs, the power outlets resides in  the site of a WSN. For the mobile ETs, the power outlet can reside in the site of a WSN or nearby the site of a WSN. By using constant power sources, the ETs in a WSN can get uninterrupted and predictable power supply, which is not the case with ambient energy harvesting WSNs.

\section{Related Works}
A fundamental property of wireless charging is that the wireless charging signal decays with the increase of distance from the ETs. This decay property of wireless charging signal is due to a fundamental law of nature, namely - the decaying property of Electromagnetic Radiation. Depending on the power of an ET, in WSN, the range of wireless signal for charging may range from 1 meter to few hundred meters, which is not enough for many applications. As a result, optimal power management of the energy transmitted by the ETs is a crucial issue in a WRFET, where the optimal power management largely depends on the placement and number of ETs and sensor nodes.

Optimal power management of WSNs involves optimal average charging of the sensor nodes and fairness of charging, which help WSNs to  achieve network goals gracefully. As mentioned earlier, the electromagnetic signal transmitted from ETs weakens over distance, the amount of energy harvested by a sensor node depends on its position relative to the position of ETs. This may results in an unbalanced energy harvesting across the network. For instance, given a WRFET, an ET node $ET_0$, two sensor node $n_0$ and $n_1$, $n_0$ will harvest more energy than $n_1$, if $distance(n_0,ET_0) < distance(n_1,ET_0)$. Therefore, the right placement of ETs is really crucial to achieve balanced energy distribution and in turn, network goals.

 In \cite{transmissionRate}, the authors have studied the joint problem of scheduling, routing and balanced power management in a WRFET and provide a solution to this problem that maximizes the\textit{ minimum average transmission rate} of the sensed data in a WSN.  Given a WSN $W$, with the set of nodes $N$, the problem of joint routing, scheduling and power management of a WSN with RFET is formulated as an optimization problem. The optimization problem maximizes the minimum average net scheduling rate $\lambda_i$ of the sensor node $i \in N$, where $\lambda_i$ = $\sum_{j\in N} t_{ij} -  \sum_{j\in N} t_{ji}$, with $t_{ij}$ be an average scheduling rate of transmissions (i.e., the average number of transmissions in a time-slot) from node $i$ to node $j$. This optimization is subject to five constraints, where three are data transmission constraints and two are energy constraints. The data transmission constraints are the followings: a) the sum of average number of transmissions between any pair of nodes in $W$ at a time-slot must not exceed 1, b) for any node $i$, the average net scheduling rate must be $\lambda_i$. c) For any node $i$, at any time-slot, the average number of transmissions in a time-slot cannot be negative. The power constraint are the following: a) The power level (from ETs) received by any node $i$  must be at least equal to the power consumed by $i$ to transmit/receive data to/from any other node $j$. b) The power needed for transmitting data from $i$ must respect the Signal-to-Noise Ratio (SNR) threshold of the intended receiver.  

In a WRFET, the mobile ETs move from one location to another location to transmit energy for the sensor nodes. The specific locations, where the ETs park, are called \textit{landmarks}. Selection of optimal landmarks play a crucial role to achieve network goals. One such network goal is to maximize profits associated with missions or tasks. In \cite{profit-maximization}, the authors have studied the problem of optimal selection of landmarks to  favour those sensors that participate in profit maximizing missions. In other words, this work solves the problem of finding the optimal number of landmarks to replenish the battery of those sensors that participate in profit maximizing missions.

Given a set of $n$ sensor nodes in a WRFET, a mission $j$ and a landmark location $(x,y)$, profit of mission $j$ is defined as: $$P_j=\sum_{i=1}^{n} \sum_{x} \sum_y z^{i}_{xy} *r_{ij} * (\sigma_{ij} / s_j) $$ where $z^{i}_{xy}=1$, if sensor $i$ receives power from a landmark at $(x,y)$, $r_{ij} =1$, if $i$ is participating in mission $j$, $\sigma_{ij}$ is the utility of sensor $i$ to mission $j$ and $s_j$ is the sensing demand of mission $j$. The utility of a sensor node $i$ is $\sigma_{ij}$ and it is defined as $1/d_{ij}$, where $d_{ij}$ is the distance distance between $i$ and $j$. Therefore, the closer a sensor is to a mission, the higher is its utility for that mission.

This optimization problem is named Mission Aware Placement of Wireless Power Transmitters (MAPIT) and is modelled as an Integer Linear Programming (ILP) problem. MAPIT jointly maximizes the followings: I) it maximizes the number of nodes receiving energy from a landmark and II) these sensors participate in profit maximizing missions.

For $m$ distinct missions in WRFET, the paper maximizes the following objective function $$maximize \;\;  \sum_{j=1}^{m} \sum_{i=1}^{n} \sum_{x} \sum_y \theta^{ij}_{xy} * (\sigma_{ij} / s_j) $$.
Where $\theta^i_{xy} = z^{i}_{xy} *r_{ij}$. $\theta^i_{xy}=1$, if $i$ is an participating sensor node in mission $j$ and receives power from ET located at landmark at $(x,y)$. Note that $\mathbf{\theta^i_{xy} = z^{i}_{xy} *r_{ij}}$, connects the node $i$ receiving energy from a landmark at $(x,y)$ with mission $j$ and$ \mathbf{(\sigma_{ij} / s_j)}$ incorporates the utility of node $i$ for mission $j$. As a result, by maximizing $\sum_{j=1}^{m} \sum_{i=1}^{n} \sum_{x} \sum_y \theta^{ij}_{xy} * (\sigma_{ij} / s_j)$, one can indeed maximize (I) and (II).

The above maximization problem is subject to the following constraints: a) Power supplied by an ET located at position $(x,y)$, must be at least equal to the power requirement of the sensor nodes that receive power from that ET (Equation 7 of \cite{profit-maximization}), b) A sensor node is associated with exactly one ET and no sensor node (Equation 8 of \cite{profit-maximization}), c) An ET located at $(x,y)$ is able to transfer power to at least one sensor node (Equation 9 and 10 of \cite{profit-maximization}), d) the number of landmarks cannot exceed a preset maximum number of landmarks (Equation 11 of \cite{profit-maximization}) e) A sensor node is associated with at most one mission (Equation 12 of \cite{profit-maximization}), f) A set of sensor nodes participating in a mission $j$, must fulfil the sensing requirement of $j$ (Equation 13 of \cite{profit-maximization}).

In \cite{fairness}, the authors have studied the optimal number and placement of static ETs for a given WRFET. To solve the optimal placement problem of ETs, the paper formulates an optimization problem wrt. a trade-off between average energy charged by the sensor nodes and fair distribution of energy among the sensor nodes. Another related optimization problem is formulated that finds the optimal number of ETs, wrt. the following constraint: minimum energy charged by each sensor nodes in a given WRFET must be at least equal to a threshold value. The work then empirically evaluates the RF energy transfer scheme, in terms of a) average energy charged by the sensor nodes and fairness of energy replenishment and b) optimal number of ETs.

To formulate the first optimization problem mentioned above, the paper considers a WRFET of $N_s$ number of sensor nodes and $N_E$ number of ETs, where $\{S_1,S_2, \dots, S_{N_S}\}$ and $\{ET_1,ET_2, \dots, ET_{N_E}\}$ are the set of sensor nodes and ETs respectively. For the sensor node $i$, $E_{C_i}$ is assumed to be the energy charged at node $i$ at any given time and $E_{C_i}$ is defined as follows
$$
E_{C_i} = \{(T_E - PL_T - P_{r,i}) * (T-\tau))\} * \mu
$$
where $T_E$ is the transmit power of $e^{th}$ ET, the$PL_T$ is the path loss in terrestrial environment, $P_{r,i}$ is the power consumed by the $i^{th}$ sensor node for power reception, $T$ is frame size (operation time+charging time), $\tau$ is the charging time within T (thus, $T-\tau$ is the charging time for a sensor node.) and $\mu$ is the charging efficiency of the sensor nodes (fixed for every node). The terrestrial path loss $PL_T$ is adopted from Very High Frequency (VHF) propagation literature and it is defined as follows:
$$
PL_T(d) = L_0(d_0) + 10 \; * w  * \; log(d/d0) + X_f
$$
where $d$ is the physical distance between transmitter and receiver, $d_0$ is the distance of a reference location from transmitter, $L(d_0)$ is the measured path loss (depends on frequency of the signal being transmitted), $w$ is the path loss exponent and $X_f$ is the Gaussian random contributor (represents the shadowing effect).

The utility  function $U$ for the $e^{th}$ ET at position $(x_e,y_e)$, that takes the average charging and fairness of energy distribution in consideration, is formulated as follow:
$$
U(x_e,y_e) = \alpha \; * \; \sum_{i=1}^{N_S} E_{C_i} + (1-\alpha)\; *\; min \; \{E_{C_i}\}
$$
Here, $0 \le \alpha \le 1$ is the trade-off factor. The \textit{first term} in the right-hand side of the above equation incorporates \textit{the total energy charged in the network} ( by the $e^{th}$ ET) and the \textit{second term} incorporates \textit{the minimum energy charged by a sensor node in the network} (by the $e^{th}$ ET). As a result, by maximizing this utility function, one can indeed maximize both average energy charging and  energy distribution fairness wrt. the trade-off factor $\alpha$. Hence, the paper solves the first optimization problem by solving the following optimization problem:

$$
{\{\bar{x_e},\bar{y_e}\}} =arg \; max_{\{x_e,ye\}}\; U(x_e,y_e)
$$
where $\{\bar{x_e},\bar{y_e}\}$ represents the optimal placement of $e^{th}$ ET.

The second optimization problem is about finding the optimal number of ETs wrt. a minimum charging requirement for any sensor node in a given WRFET. This problem is formulated as follows:

$$
\bar{N_E} = arg\;min\;\{N_E \vert E_{C_i} \ge \eta\}
$$
where, $\bar{N_E}$ is the optimal number of energy transmitter and $\eta$ is the minimum energy charging requirement for any sensor node. An acceptable error margin $\zeta$ is defined as follows: $\zeta = \eta - \{E_{C_i}\}$. While solving the equation for $\bar{N_E}$, a trade-off between acceptable error margin and required number of energy transmitters is considered.


\section{Problem Formulation}
\label{problemFormulation}
\subsection{Task and Energy Aware Node Placement}
Figure 1 shows an example layout of a WSN, where a Task and Energy Aware Node Placement (TENP) problem can be formulated. In this example layout, we have 7 sensor nodes participating in  3 different tasks ($T_1,T_2$ and $T_3$) and receiving energy from 3 ETs. 

Given a WSN, we consider a set of sensor nodes $N=\{s_1, \dots s_m\}$, a set tasks $\Gamma=\{t_1 \dots t_n\}$, a set of utility requirements of the tasks $U=\{u_1 \dots u_n\}$ (where, $u_i \in U$ corresponds to the utility requirement of the task $t_i \in \Gamma$) and a set of ETs, $E=\{e_1 \dots e_o\}$. For each tasks $t_j$, we assume a set of sensors $N^{t_j} \subseteq N$, with $N^{t_j}=\{ s^{t_j} \vert s^{t_j} \;participates \; in \; task\; t_j\}$ .  We require that a sensor node is associated with exactly one task, i.e., $\forall_{{t_i}{t_j}}$, with $i \neq j$, $ N^{t_i}  \cap N^{t_j} =\emptyset$,  and  $\bigcup_{i=1}^{n} N^{t_i} = N$. We also assume $C_{i}^{k}$ be the charge received by the sensor node $s_i \in N$ from  energy transmitter $e_k \in E$ within a given time frame and $\lambda$ be the charging requirement of any sensor node $s_i \in N$.
\begin{figure}
   \centering
   \begin{minipage}{0.4\linewidth}
   		\label{1CPUTime}
       \centering
       \includegraphics[width=1\linewidth]{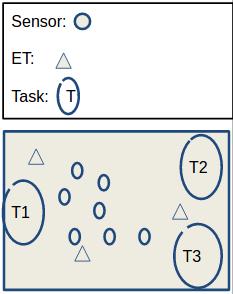} 
       \caption{An environment for TENP}
   \end{minipage}
\end{figure}
In the literature (such as \cite{profit-maximization}), utilization of a sensor node wrt. a task is defined in terms of the distance between that task and that sensor. Precisely, for a task $t_j$, we define the utilization $\mu_{s^{t_j}_i}$ of a node $s^{t_j}_i \in N^{t_j}$, wrt. $t_j$, as follows:  $\mu_{s^{t_j}_i}=1/(distance(s^{t_j}_i,t_j))$ \footnote{In this paper, we assume Manhattan distance as our \textit{distance} function, where $distance(x,y)=\vert x - y \vert$.}. So, the closer (further) a sensor is placed from its associated task, the more (less) is its utility for that task.

The TENP problem is formulated as follows:
$$\!\!\!\!\!\!\!\!\!\!\!\!\!\!\!\!\!\!\!\!\!\!\!\!\!\!\!\!\!\!\!\!\!\!\!\!\!\!\!\!\!\!\!\!\!\!\!\!\!\!\!\!\!\!\!\!\!\!\!\!\!\!\!\!\!\!\!\!\!\!\!\!\!\!\!\!\!\!\!\!\!\!\!\!\!\!\!\!\!\!\!\!\!\!\!\!\!\!\!\mathbf{Minimize}$$
$$
  \{\sum_{j=1}^{n}\sum_{i=1}^{\vert N^{t_j}\vert} distance(s^{t_j}_{i},t_j) \} + \{\sum_{k=1}^{o}\sum_{i=1}^{\vert N^{t_j}\vert} distance(s^{t_j}_i,e_k) \} \;{\forall_{t_j \in \Gamma}} ,
$$
$$
\!\!\!\!\!\!\!\!\!\!\!\!\!\!\!\!\!\!\!\!\!\!\!\!\!\!\!\!\!\!\!\!\!\!\!\!\!\!\!\!\!\!\!\!\!\!\!\!\!\!\!\!\!\!\!\!\!\!\!\!\!\!\!\!\!\!\!\!\!\!\!\!\!\!\!\!\!\!\!\!\!\!\!\!\!\!\!\!\!\!\!\!\!\!\!\!\!\!\! \mathbf{Subject} \; \mathbf{to} \;
$$
$$
\!\!\!\!\!\!\!\!\!\!\!\!\!\!\!\!\!\!\!\!\!\!\!\!\!\!\!\!\!\!\!\!\!\!\!\!\!\!\!\!\!\!\!\!\!\!\!\!\!\!\!\!\!\!\!\!\{\sum_{k=1}^{o} C_{i}^{k} \ge \lambda \},  \;\;\; \forall_{s_i \in S}
$$
$$
\!\!\!\!\!\!\!\!\!\!\!\!\!\!\!\!\!\!\!\!\!\!\!\!\!\!\!\!\!\!\!\!\!\!	  \{\mu_{s^{t_j}_i} \ge u_j\}, \;\;\; \forall_{s^{t_j}_{i} \in N^{t_j}}, \forall_{t_j \in \Gamma}
$$
In this paper, we solve this optimization problem.

\section{Algorithms}
In this section, we devise an algorithmic framework to solve the TENP problem. First, we establish some required notations to present the algorithms, then we present the algorithms.
\subsection{Notations}
We consider a $n$ by $n$ square grid of cells as the sensor node deployment \textit{Environment}, where the location of a cell is identified by a Cartesian coordinate. We define the environment $\varepsilon$ as follows:
$$
\varepsilon=\{(x_1,y_1), \dots, (x_i,y_i) , \dots , (x_n,y_n)\}
$$, where $(x_i,y_i)$ is the location of the $i^{th}$ cell in $\varepsilon$. For each of the energy transmitters $e \in E$, a cell in $\varepsilon$ is reserved. We define $\varepsilon^E \subset \varepsilon$ to be the set of cells reserved for energy transmitters. Similarly, for each of tasks $t \in \Gamma$, a cell in $\varepsilon$ is reserved. We define $\varepsilon^\Gamma \subset \varepsilon$ to be the set of cells reserved for tasks. We also have  $\varepsilon^E  \cap \varepsilon^{\Gamma} = \emptyset$. The sensor nodes in $N$ can be placed in  $\varepsilon^N \subset \varepsilon$, with $\varepsilon^N = \varepsilon \setminus (\varepsilon^E  \cup \varepsilon^\Gamma)$ and $\vert \varepsilon^N \vert \geq \vert N \vert$.

For each of the tasks $t_i \in \Gamma$, we have a set of sensor nodes $N^{t_i}$, which are associated with task $t_i$. Thus, we map the task/sensor association in
$N^{\Gamma}=\{(t_1, N^{t_1}),\dots , (t_i, N^{t_i}), \dots (t_n, N^{t_n})\}$. Additionally, we denote the task $t$ is associated with the sensor $s$ as $t^s$.

\subsection{Distance Minimzation}

Algorithm 1 shows the algorithm of the procedure \textit{distanceMinimization}, which solves a given TENP problem. The idea of the \textit{distanceMinimization} procedure is as follows: (a) assign a sensor $s$ a location $(x,y)$ from the available locations in $\varepsilon^N$, for which the charging requirement constraint of $s$ and utility constraint of task $t^s$ can be satisfied and the location $(x,y)$ is distance minimized wrt. $t^s$ and all the ETs $e \in E$. Perform (a) repeatedly for all the sensors $s \in N$. If every sensor $s \in N$ gets a position, then the problem is satisfiable, otherwise it is not.

The main loop of Algorithm 1 iterates through all the tasks $t_i \in \Gamma$ (line 1). Then, for each sensor  $s^{t_i} \in N^{t_i}$ (line 2), for each available location $(x,y)$ (line 5),  it computes the distance $taskDistance$ of $(x,y)$ from task location $(x_{t_i}, y_{t_i})$ (line 6). It also computes the combined distance $etDistance$ of  $(x,y)$ from all the ET locations $(x_e,y_e) \in \varepsilon^{E}$ (line 7-10). Then for the location $(x,y)$, it computes the \textit{combinedDistance} by summing up $taskDistance$ and $etDistance$ (line 11). Once, combined distance for all the available locations $(x,y)$ is computed, it sorts the potential locations for $s^{t_i}$ by increasing order of $combinedDistance$ (line 13) and place the sorted results in $sortedIndexedDistance$. At this point, for each locations $(x',y,) \in sortedIndexedDistance$, Algorithm 1 checks if both of the utility constraints and charging constraints can be satisfied (line 16-17). If both of these constraint can be satisfied for the location $(x',y')$, then it means that the sensor $s^{t_i}$ can be placed in $(x',y')$, in which case, $s^{t_i}$ is placed in $(x',y')$ (line 19) and the location $(x',y')$ is marked as unavailable (line 20). If none of the available locations $(x',y') \in sortedIndexedDistance$ can satisfy both of the constraints, then it means that  $s^{t_i}$ cannot have a place\footnote{If a sensor cannot be placed in a position (a failure), then the given TENP problem is unsatisfiable. At this stage, one approach would be to terminate execution of Algorithm 1 by reporting the unsatisfiability of the given problem. But, in our algorithm, we choose to continue despite of this failure. This is because of our intention of getting solution (partial), even for the unsatisfiable problems. (line 28 of Algorithm 1)}.

At the end, the given problem is satisfiable, if every sensor node  $s \in N$ gets a position (line 25) and it is unsatisfiable, otherwise (line 28).
\begin{algorithm}

		\label{aluScore}
			\caption{distanceMinimization}
			\begin{algorithmic}[1]
			   \scriptsize
				\FORALL{$t_i \in \Gamma$}
					\FORALL{$s^{t_i} \in N^{t_i}$}
						\STATE $locIndex = 0$
%
%
%
						\STATE {$(x_{t_i},y_{t_i}) \leftarrow loc(T,t_i)$}
						\FORALL{$(x,y) \in \varepsilon^{N} \setminus \{\infty\}$}
							\STATE {$taskDistance(locIndex) \leftarrow \vert (x-x_{t_i})\vert + \vert (y-y_{t_i})\vert$}
							\STATE {$etDistance(locIndex) \leftarrow 0$}
							\FORALL{$(x_e,y_e) \in \varepsilon^{E}$}
								\STATE {$etDistance(locIndex) \;\;\;\;\; \leftarrow  etDistance(locIndex) +\vert (x-x_{e})\vert + \vert (y-y_{e})\vert$}
							\ENDFOR
							\STATE {$combinedDistance(locIndex) \leftarrow taskDistance(locIndex) + etDistance(locIndex)$ }
						\ENDFOR
						\STATE {$sortedIndexedDistance \leftarrow sort(combinedDistance)$ }
					\ENDFOR
					\FORALL{$\{d,(x',y')\} \in sortedIndexedDistance$}
						\STATE {$satisfy \leftarrow checkUitlityConstraints((x',y'),(x_{t_i},y_{t_i}))$}
						\STATE {$satisfy \leftarrow checkChargingConstraints((x',y'), \{(x'_e,y'_e) \vert (x'_e,y'_e) \in loc(E,e) \;\;\; and \;\;\; \forall_{e \in E} \})$}
						\IF {$satisfy$}
							\STATE $position(s^{t_i}) \leftarrow (x',y')$		
							\STATE $\varepsilon^{N}(x',y') \leftarrow \infty$					
							\STATE break
						\ENDIF
					\ENDFOR
				\ENDFOR
				\IF {$\vert position \vert = \vert N \vert$}
					\STATE Return \{$position$, SATISFIABLE\}						
				\ELSE
					\STATE Return \{$position$, UNSATISFIABLE\}
				\ENDIF
			\end{algorithmic}
\end{algorithm}

\subsection{Constraint Satisfaction: Utility and Charging Requirement}
Line 16 and 17 of Algorithm 1 calls two procedures, namely - \textit{checkUtilityConstraints} and \textit{checkChargingConstraints} respectively. Given a location $(x,y)$ and a sensor $s$, \textit{checkUtilityConstraints} checks if the utility requirement set by $t^s$ can be satisfied by the location $(x,y)$ and \textit{checkChargingConstraints} checks if charge received at $(x,y)$ satisfies the charging requirement $\lambda$. In this section, we present the algorithms for these two procedures.

\paragraph{Utility Constraint}The algorithm for the \textit{checkUitlityConstraints} is shown in Algorithm 2. Given a potential location  $(x,y)$ for a sensor $s$ and the location $(x_{t_i},y_{t_i})$ of it's task $t_i$, it simply checks if the utility associated for the location $(x,y)$ is at least equal to the utility requirement of the task $t_i$ (line 2). 
\begin{algorithm}[H]
		\label{aluScore}
			\caption{$checkUtilityConstraints((x,y),(x_{t_i},y_{t_i}))$}
			\begin{algorithmic}[1]
			   \scriptsize
   				\STATE $sTDistance \leftarrow \vert (x-x_{t_i}) \vert + \vert (y-y_{t_i}) \vert $ 
   				\STATE Return  $u_i \le (1/sTDistance)$
			\end{algorithmic}
\end{algorithm}

\paragraph{Charging Requirement Constraint}
Before presenting the algorithm for the procedure $checkChargingConstraints$, we construct a charge receiving model for the sensor nodes.

During a time frame $T$, the amount of charge received by a sensor node from an ET depends on many factors, such as its distance from that ET, obstacles between the sensor node and the ET, energy emitted by the ET, charging efficiency of the sensor node circuit etc. Here, we adopt the charging model described in \cite{fairness}.

For the sensor node $i \in N$, $C^{e}_{i}$ is the energy charged at node $i$ with charge transmitted from ET $e$ within the time frame $T$ and $C^{e}_{i}$ is defined as follows
$$
C^{e}_{i} = \{(T_E - PL_T - P_{r,i}) * (T-\tau))\} * \mu
$$
where $T_E$ is the transmit power of $e^{th}$ ET, $PL_T$ is the path loss in terrestrial environment, $P_{r,i}$ is the power consumed by the $i^{th}$ sensor node for power reception, $T$ is the frame size (operation time+charging time), $\tau$ is the charging time within T (thus, $T-\tau$ is the charging time for a sensor node.) and $\mu$ is the charging efficiency of the sensor nodes (fixed for every node).
The terrestrial path loss $PL_T$ has various models, such as, \textit{log-distance path loss}, \textit{log-normal shadowing} \cite{path-loss}. Here we is consider the \textit{log-distance path loss} model. In this model, the path loss $PL_T$  is defined as follows:
$$
PL_T(d) = PL_0(d_0) + 10 \; * w  * \; log(d/d0)
$$
where $d$ is the physical distance between transmitter and receiver, $d_0$ is the distance (value can be up to 100m, depending on the application) of a reference location from the transmitter, $PL(d_0)$ is free-space path loss (depends on frequency of the signal being transmitted), $w$ (value varies from 2 to 6) is the path loss rate.

The free-space path loss $PL_0(d_0)$ wrt. the distance $d_0$ of a reference location is defined \footnote{Taken from https://en.wikipedia.org/wiki/Free-space\_path\_loss} as follows:
$$
	PL_0(d_0) = ({\dfrac{4\pi d_{0}f}{c}})^2 = 20 \; log_{10} (d_0) + 20 \; log_{10} (f) + 92.5
$$, where $c$ is the speed of light and $f$ is the frequency of the transmitted signal by an ET.

In Algorithm 3, we present the procedure $checkChargingConstraints$. Given a potential location $(x,y)$ for a sensor and positions of energy transmitters $etPositions$, it computes (line 2-7) the charge received at $(x,y)$ from each of the ET $e$, located at $(x_e,y_e)$. In line 5, it incrementally sums up charged received at $(x,y)$ from each ET location $(x_e,y_e)$. Then it simply checks if $C_{(x,y)}$, the total charged received at location $(x,y)$ is at least equal to the charging requirement $\lambda$ (line 8).
\begin{algorithm}[H]
		\label{aluScore}
			\caption{$checkChargingConstraints((x,y),etPositions)$}
			\begin{algorithmic}[1]
			   \scriptsize
   				\STATE {$C_{(x,y)} = 0$}
 				\FORALL{$(x_e,y_e) \in etPositions$}
 					  \STATE $d \leftarrow \vert (x-x_{e}) \vert + \vert (y-y_{e}) \vert $ 
       				  \STATE $PL_0(d_0) \leftarrow 20 * log(d_0) + 20 * log(f) + 92.5$;
       				  \STATE $PL(d) \leftarrow PL_0(d_0) + 10 * w * log(d/d_0)$;
					  \STATE $C_{(x,y)} \leftarrow C_{(x,y)} + \{(T_E - PL_T - P_{r,i}) * (T-\tau))\} * \mu $
 				\ENDFOR
				\STATE Return $ \lambda \le C_{(x,y)}$  
			\end{algorithmic}
\end{algorithm}

\subsection{Simulation}
In Algorithm 1, for satisfiable problems, assigned positions of all the sensor nodes are returned in $position$ (line 26). With locations of all the sensor nodes being fixed in $position$, simulation can be performed by placing sensor nodes in locations as specified in $position$. In our simulation algorithm, we assume the following:
\begin{itemize}
\item During a \textit{time frame} $T$, a sensor node performs two operations: (a) Network operation (sensing and data transmission)  and (b) Charging operation, where (a) is performed for $\tau$ time  and (b) is performed for ($T-\tau$) time.
\item The total simulation time is a multiple of the given time frame size. 
\item Charging of all the sensors occurs at the beginning of a time frame for a fixed fraction of the the time frame size.
\end{itemize}

The pseudo-code for the $simulation$ procedure is shown in Algorithm 4. Starting at time 0, the $simulation$ runs until $simulationTime$ is reached (line 2-17). At beginning of each time frame (line 4\footnote{If $reminder(timeFrac,1) == 0$, then current time index $i$ is the beginning of a new time frame, where $timeFrac \leftarrow i/T$.}), all the sensor nodes receives charge for ($T-\tau$) amount of time (line 6-13). For the rest of the time $\tau$, network operations are performed (line 15).
When simulation ends at $simulationTime$, we compute the following two evaluation metrics:
\begin{itemize}
\item \textbf{Average Harvested Charge}: The average amount of charge that the sensor nodes harvest from the energy transmitters within a time frame. Thus, average harvested charge of a given WSN is $\frac{\sum_{i=1}^{\vert N \vert}\sum_{k=1}^{\vert E \vert} C_{i}^{k}}{\vert N \vert * T}$ (Computed in line 19-23).
\item \textbf{Average Task Utility}: The average amount of utility yield by tasks from their associated sensors within the simulation life-time . Thus, the average task utility of a given WSN is $\frac{\sum_{j=1}^{\vert \Gamma \vert}\sum_{i=1}^{\vert N^{t_j} \vert} \mu_{s^{t_j}_i}}{\vert \Gamma \vert}$ (Computed in line 23).
\end{itemize}.

\begin{algorithm}[H]
		\label{aluScore}
			\caption{$simulation(position,simulationTime,T)$}
			\begin{algorithmic}[1]
			   \scriptsize
   				\STATE $i \leftarrow 0$
   				\WHILE {$i < simulationTime$}
   					 \STATE $timeFrac \leftarrow i/T$
   				     \IF{$reminder(timeFrac,1) == 0$}
						\STATE $i'\leftarrow 0$  
						\WHILE {$i' <= (T-\tau) $}
							\STATE $s=0$	
							\FORALL{$(x,y) \in position$}
 					  			\STATE $recievedCharge(s) \leftarrow computeRecievedCharge(x,y)$
 					  			\STATE $s \leftarrow s+1$
 							\ENDFOR	
 							\STATE $i' \leftarrow i'+1$					
						\ENDWHILE
   				     \ENDIF
   				     \STATE $performNetworkOperation()$			    
   				     \STATE $i \leftarrow i+1$					
   				\ENDWHILE
   				\STATE $j \leftarrow 0$					
   				\FORALL{$t \in T$}
 					  	\STATE $utility(j) \leftarrow computeUtility(t,N^t)$
 					  	\STATE $j \leftarrow j+1$
 						\ENDFOR	
 				\STATE Return \{$averageCharge \leftarrow mean(recievedCharge)/T$, $averageUtility \leftarrow mean(utility)$\}
			\end{algorithmic}
\end{algorithm}

\section{Theoretical Analysis} 
In this section, we present theoretical analysis of our main algorithm, namely - the \textit{distanceMinimization} (Algorithm 1) algorithm. First, we analyze \textit{time complexity} of Algorithm 1. Then we present our discussion on the \textit{completeness} of Algorithm 1 w.r.t the TENP problem.

\subsection{Time-complexity of Algorithm 1}
\begin{prop}
The worst case time-complexity of Algorithm 1 is $O(\vert N \vert^{2}* (\vert E \vert +  log (\vert N \vert) + \frac{\vert E \vert}{\vert N \vert}))$.
\end{prop}
\begin{proof}
We develop the time-complexity of Algorithm 1 by dividing its execution into the following divisions:
\begin{itemize}
\item [(a)] The outermost loop of Algorithm 1 (line 1) runs $\vert \Gamma \vert$ number of times and the second outermost loop of Algorithm 1 (line 2) runs for $\vert N^{t_i} \vert$ for a task $t_i \in \Gamma$.  So, taken together, the first and second outermost loops run for total $\sum_{i=1}^{\vert \Gamma \vert}  N^{t_i}$ number of times. As our problem formulation assumes that a sensor node is associated with exactly one task (i.e., $\forall_{{t_i}{t_j}}$, with $i \neq j$, $ N^{t_i}  \cap N^{t_j} =\phi$,  and  $\bigcup_{i=1}^{n} N^{t_i} = N$), we have $\vert N \vert = \sum_{i=1}^{n}\vert  N^{t_i} \vert $. As a result, taken together, the two outermost  loops run for $\vert N \vert$ number of times. 

For a given TENP problem,\textit{ worst case} occurs when no sensors can be placed in any of the cells of $\varepsilon^N$. In this case, no position in $\varepsilon^N$ is marked as occupied (i.e., marked by $\infty$) at any stage of the execution. And as a result, for each of the  $\vert N \vert$ iterations (of the two outermost loops, taken together), we have $\vert \varepsilon^{N} -\{\infty\} \vert = \vert N \vert$. Then, in each of these $\vert N \vert $ iterations (of the two outermost loops, taken together), in worst case, we have $\vert N \vert $ iterations of division (1) and (2). That is, in worst case, division (1) and(2) runs for $\vert N \vert^{2}$ times.
\begin{itemize}

\item [(1)] For each execution of the loop at line 5, the distance from each of energy transmitters $e \in E$ is computed (line 8). Thus, in each of $\vert N \vert^2$ invocations, this loop runs exactly $\vert E \vert$ steps .
\item [(2)]  The call to the sort function is not a constant time operation. Here, we assume the best sorting algorithm for which the worst case run-time is \textit{log} of the size of the input. Thus in worst case, the complexity of the sort function is  $log (\vert N \vert)$. In the worst case, the sort function sorts $|N|$ elements in each of the $\vert N \vert^2$ invocations. 

\end{itemize}
\item [(b)]  Algorithm 1 checks the utility and charging requirement constraints for all the available locations in $sortedIndexedLocations$ (has the same size of  $\varepsilon^{N} -\{\infty\}$).  So, these two checks are performed at most $\vert N \vert$ times in the worst case.
\begin{itemize}
\item [(1)]  Algorithm 1 makes a call to the procedure $checkChargingConstraints$ for each sorted locations, which executes a loop over all $e \in E$. So, for this loop the algorithm incurs another $\vert E \vert$ executions in each invocations .
\end{itemize}
\end{itemize}

We have the following cost of execution per divisions: For (a-1), it incurs cost of $O(\vert E \vert)$. For (a-2), it incurs cost of $O(log (\vert N \vert)$.  $O(\vert N \vert)$ and $O(\vert N \vert * \vert E \vert)$ is the execution cost for (b) and (b-1) respectively. As mentioned earlier,  (1) and (2) executes $\vert N \vert^2$ number of times and (b) and (b -1) runs for $\vert  N \vert$ and $\vert  N \vert * \vert  E \vert$ number of times respectively. Thus the worst case running time for Algorithm 1 is
$$O(\vert N \vert^2 *(\vert E \vert +  log (\vert N \vert)) + $$
$$O(\vert N \vert) + $$
$$O(\vert N \vert * \vert E \vert)$$
$$
=O(\vert N \vert^{2} * \vert E \vert + \vert N \vert^{2} * log (\vert N \vert) + \vert N \vert * \vert E \vert )
$$
$$
=  O(\vert N \vert^{2}* \;\;(\vert E \vert +  log (\vert N \vert) + \frac{\vert E \vert}{\vert N \vert}))
$$
\end{proof}
\subsection{Completeness of Algorithm 1}
\begin{prop}
Algorithm 1 is an incomplete but tractable method for solving the TENP problem.
\end{prop}
\begin{proof}

Let $s^{t_i} \in N^{t_i}$ was assigned to a position $(x,y)$ at $v^{th}$ execution of the outermost loop (line 1) of of Algorithm 1 (successful case). Suppose, at the $v^{' th}$ (where $v^{'} > v$) execution of the outermost loop (line 1) of Algorithm 1, position assignment of $s^{t_k} \in N^{t_k}$ fails. This failure is due to the failure to satisfy either of (1) utility requirement of the task $t_k$ or (2) the charging requirement of the sensor $s^{t_k}$. However, the previous assignment of $s^{t_i}$ at position $(x,y)$, can be a reason of the failure of (1) or (2). Because, an alternative location $(x',y')$ may exist for $s^{t_i}$, assignment of which to $s^{t_i}$ could potentially avoid the failure of the position assignment of $s^{t_k}$. In other words, changing a previous sensor node assignment, may create a position (where both of the constraints are satisfied) for the currently failed sensor node. Algorithm 1, does not consider this possibility.

As a consequence, in general, Algorithm 1 is not able to determine the satisfiability of all the \textit{satisfiable} instances. In other words, Algorithm 1 is an incomplete method for solving the TENP problem. 

From Theorem 1, it is clear that it runs in polynomial time in the worst case. Putting together, Algorithm 1 is an incomplete but tractable method for solving the TENP problem.
\end{proof}


\section{Empirical Evaluation}
We have empirically evaluated the TENP problem by implementing the algorithms described in Section 5 and then by performing some experiments with the implementation. For our implementation, we have used MATLAB. Our implementation code-base is spread across 4 scripts and approximately has 200 lines of code.
\subsection{Test Bed for Experiment}
We set up our test bed for experimentation by setting fixed values for two types of parameters: (i) Energy, time frame and charging parameters and (ii) Environment parameters. Table 1 shows the values for type (i) parameters and Table 2 shows values for type (ii) parameters. Most of the parameter values in Table 1 are adopted from \cite{fairness} and parameter values of Table 2 are arbitrarily chosen to produce a WSN that closely mimic a realistic WSN.

\begin{table}[H]
\centering
\makebox[0pt][c]{\parbox{0.5\textwidth}{%
   \begin{minipage}[b]{\hsize}\centering
   \begin{tabular}{|l|l|l|}
\hline
\textbf{Parameter}                                                                                & \textbf{Symbol} & \textbf{Value}                \\ \hline
Transmission Power of ET E                                                                        & \textit{T\_E}   & 50 dBm (100 Watt)  \\ \hline
Frequncy                                                                                          & \textit{f}      & 2 GHz \\ \hline
Path loss rate                                                                                    & \textit{w}      & 2                           \\ \hline
Reference Location Distance                                                                       & $d_0$            & 5m                      \\ \hline
Charging Circuit Efficiency                                                                       & $\mu$             & 0.5                    \\ \hline
Time Frame Size                                                                                   & T               & 10 S                          \\ \hline
\begin{tabular}[c]{@{}l@{}}Operation Time for Sensor \\ Node\end{tabular}                         & $\tau$            & 9.5 S                           \\ \hline
\begin{tabular}[c]{@{}l@{}}Power Consumption of Sensor \\ Nodes for Charge Reception\end{tabular} & $p_{r,i}$      & 30 dBm (1 Watt)                        \\ \hline
\end{tabular}
       \caption{Energy, Time Frame and Charging Parameters}
       \label{tab:singlebest}
   \end{minipage}
   \hfill
   \begin{minipage}[b]{1\hsize}\centering
       \begin{tabular}{|l|l|l|}
\hline
\textbf{Input Parameter}  & \textbf{Symbol}                      & \textbf{Value}                     \\ \hline
Environment Size          & \textit{$ \vert \varepsilon \vert $} & 10 * 10                            \\ \hline
Number of Tasks           & \textit{$ \vert T \vert $}           & 4                                  \\ \hline
Number of ETs             & \textit{$ \vert E \vert $}           & 4                                  \\ \hline
Number of Sensors         & $ \vert N \vert $                    & 26                                 \\ \hline
Tasks Position            & $\varepsilon^T$                      & {[}(1,10),( 10,3),(10,1),( 1,1){]} \\ \hline
ET position               & $\varepsilon^E$                      & {[}(9,1),( 8,5),( 6,8),( 6,5){]}   \\ \hline
Tasks-Sensors Association & $N^{t_1}$                            & (2, 3, 4, 8, 10, 15, 16, 17)       \\ \hline
       & $N^{t_2}$                            & (5, 7, 11, 18, 19, 20, 21, 22)     \\ \cline{2-3}
                         & $N^{t_3}$                            & (6, 9, 12, 13, 14, 23, 24)         \\ \cline{2-3}
                         & $N^{t_4}$                            & (25, 26, 27)                       \\ \hline
\end{tabular}
       \caption{Environment Parameters}

   \end{minipage}
}}
\end{table}

\subsection{Experimental Results}
In this section, we report our experiments and experimental results of the implemented algorithms.

For experimentation, we consider three versions of TENP:
\begin{itemize}
\item \textbf{TENP}: We compute the placement of sensor nodes wrt. both of the task utility constraint and charging constraint. Note that in this experimental settings, we consider the full version of TENP.
\item \textbf{TSP}:  TSP is a relaxed version of TENP, where only the task utility requirement constraint is considered. That is, in TSP, we disregard the charging requirement constraint.
\item \textbf{ESP}: ESP is a relaxed version of TENP, where only the charging requirement constraint is considered. That is, in ESP, we disregard the task utility requirement constraint.
\end{itemize}
The purpose of considering TSP and ESP is this: we want to compare the \textit{average task utility} and \textit{average harvested charge} for full version of TENP and these relaxed versions of TENP (Reported in Table 3).

\subsubsection{Experiment with TENP}
With this full version of TENP, we perform two sets of experiments under the following two experimental settings. Given a TENP problem $P$,
\begin{itemize}
\item [(I)] We vary charging requirement ($\lambda$) with fixed task utility requirement ($u$), until $P$ can be satisfied with $(\lambda,u)$.
\item [(II)] We vary task utility requirement ($u$) with fixed charging requirement ($\lambda$), until $P$ can be satisfied with $(\lambda,u)$.
\end{itemize} 

For a given pair of values of $(\lambda,u)$, first we determine the positions of the sensor nodes in $\varepsilon$ by using the \textit{distanceMinimization} script (implementation of Algorithm 1) and then use the determined positions in \textit{position} to simulate charging and network operations using the \textit{simulation} script (implementation of Algorithm 4).

\paragraph{TENP with Setting (I)}Figure 2 and Figure 3 respectively shows the change of \textit{average harvested charge} and \textit{average task utility} as the value of $\lambda$ (charging requirement) increases from 1 to 17 (with $\lambda=18$, the problem becomes unsatisfiable). We assume fixed values for task utility requirement, with $u=[0.2,0.2,0.2,0.2]$.
\begin{itemize}
\item We have two observations: As we increase $\lambda$, (a) average harvested charge never decreases (Figure 2) and  (b) average task utility never increases (Figure 3). This is an expected result, with the increase of charging demand,  the average harvested charge increases (and the average task utility decreases as task utility requirement remains constant).

\begin{figure}[H]
   \centering
   \begin{minipage}{0.8\linewidth}
   		\label{1CPUTime}
       \centering
       \includegraphics[width=1\linewidth]{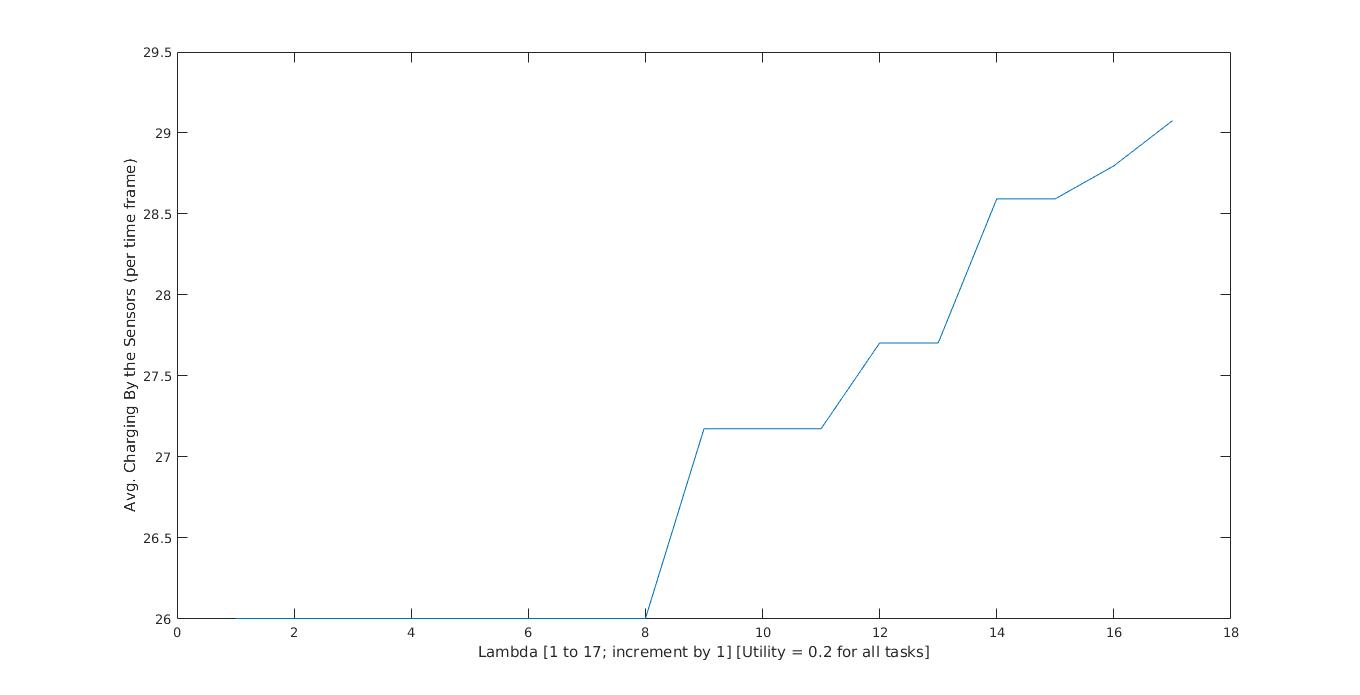} 
       \caption{Average Harvested Charge for TENP with varying $\lambda$ and fixed $u$}
   \end{minipage}
   \begin{minipage}{0.8\linewidth}
       \centering
       \label{1ratio}
       \includegraphics[width=1\linewidth]{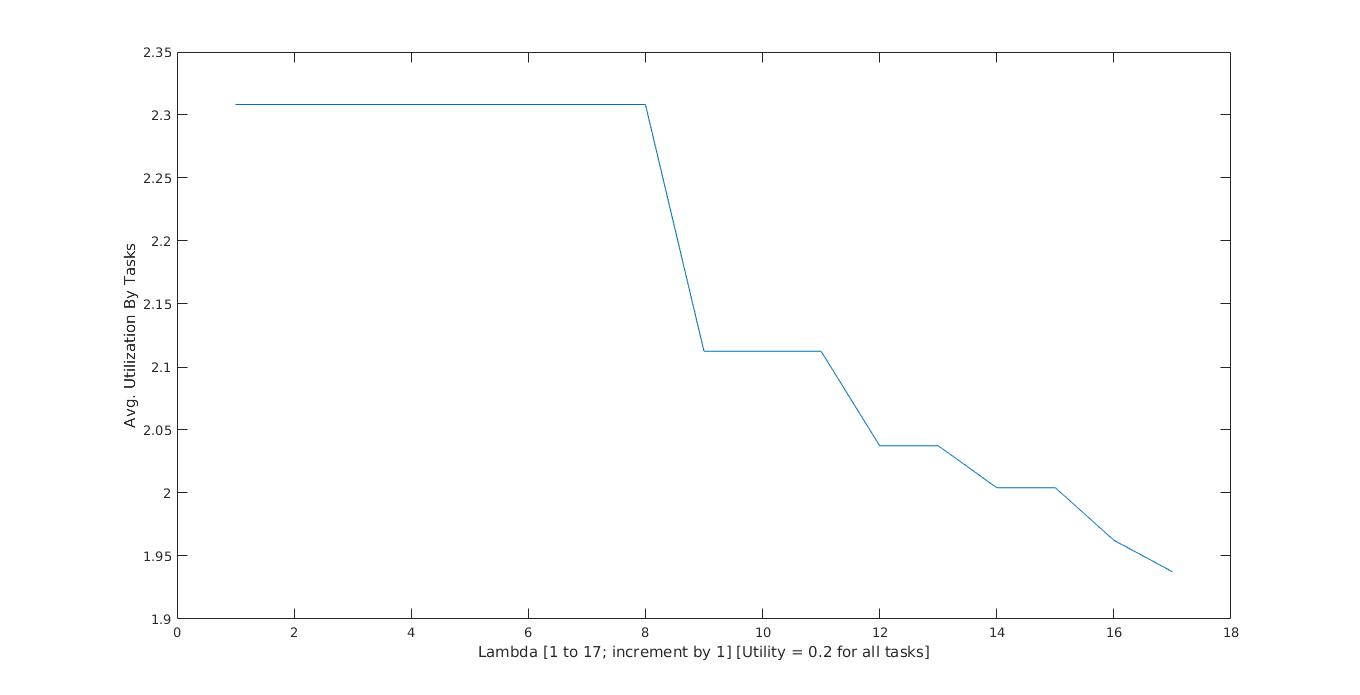} 
       \caption{Average Task Utility for TENP with varying $\lambda$ and fixed $u$}
   \end{minipage}
\end{figure}
\item We notice the presence of symmetry between the average harvested charge and average task utility in Figure 2 and 3. Without loss of generality, for TENP, average harvested charge and average task utility appears to be \textit{almost inversely proportional} to each other, wrt. $\lambda$. 

\item In both Figure 2 and Figure 3, we observe the \textit{stalling phenomenon}. The average harvested charge (resp. average task utility) does not always increase (resp. decrease) with the increase of $\lambda$, instead for some consecutive values of $\lambda$ (for example, $\lambda$ values from 1 to 8), both of average harvested charge and average task utility remains unchanged. The explanation of this stalling phenomenon is this: Let, $\lambda^{'}$ and $\lambda^{''}$ are two values of charging requirement within a stalled region, with $\lambda^{'} < \lambda^{''}$. For this particular scenario, for both $\lambda^{'}$  and $\lambda^{''}$, the problem can be satisfied with the sensor nodes placed in the same positions. As positions of the sensor nodes do not change with the increase of $\lambda$ value within a stalled region, the average harvested charging and average task utility do not change.
\end{itemize}
\paragraph{TENP with Setting (II)}
Figure 4 and Figure 5 respectively shows the change of \textit{average harvested charge} and \textit{average task utility} as task utilities in $u$ are increased from \textit{(0.1,0.1,0.1,0.1)} to\textit{ (0.25,0.25,0.25,.025)}, where in each increment, we increase the values in $u$ by 0.01. (with $u=(0.26,0.26,0.26,.026)$, the problem becomes unsatisfiable). We assume a fixed value for $\lambda$ (=10).
\begin{itemize}
\item Like the setting (I), we have two observations: as we increase $u$, (a) average harvested charge \textit{almost always} decreases (Figure 4) and  (b) average task utility \textit{almost always} increases (Figure 5). This is an expected result, as increasing task utility demands, increases the average task utility (and decreases average harvested charging, with fixed charging requirement of the sensors).
\item In Figure 4 and Figure 5, we also notice the presence of symmetry between the average harvested charge and average task utility. Without loss of generality, for TENP, average harvested charging and average charge utility appears to be \textit{almost inversely proportional} to each other, wrt. the task utilization requirement $u$. 
\item Like setting (I), for both Figure 4 and Figure 5, we observe the \textit{stalling phenomenon}. The average task utility  (resp. average harvested charge) does not always increase (resp. decrease), with the increase of $u$, instead for some consecutive values of $u$ (for example, for $u$ values from 0.17 to 0.20), both of average harvested and average task utility remains unchanged. The explanation of this stalling phenomenon is same as what it is for setting (I): Let, $u^{'}$ and $u^{''}$ are two values of charging requirement within a stalled region, with $u^{'}_i < u^{''}_i$, for $1 \le i \le 4$. For this particular scenario, for both $u^{'}$  and $u^{''}$, the problem can be satisfied with the sensor nodes placed in the same positions. As position of the sensor nodes do not change with the increase of values in $u$, the average task utility and average harvested charge do not change.
\end{itemize}

\begin{figure} [H]
   \centering
   \begin{minipage}{0.8\linewidth}
       \centering
       \label{2CPUTime}
       \includegraphics[width=1\linewidth]{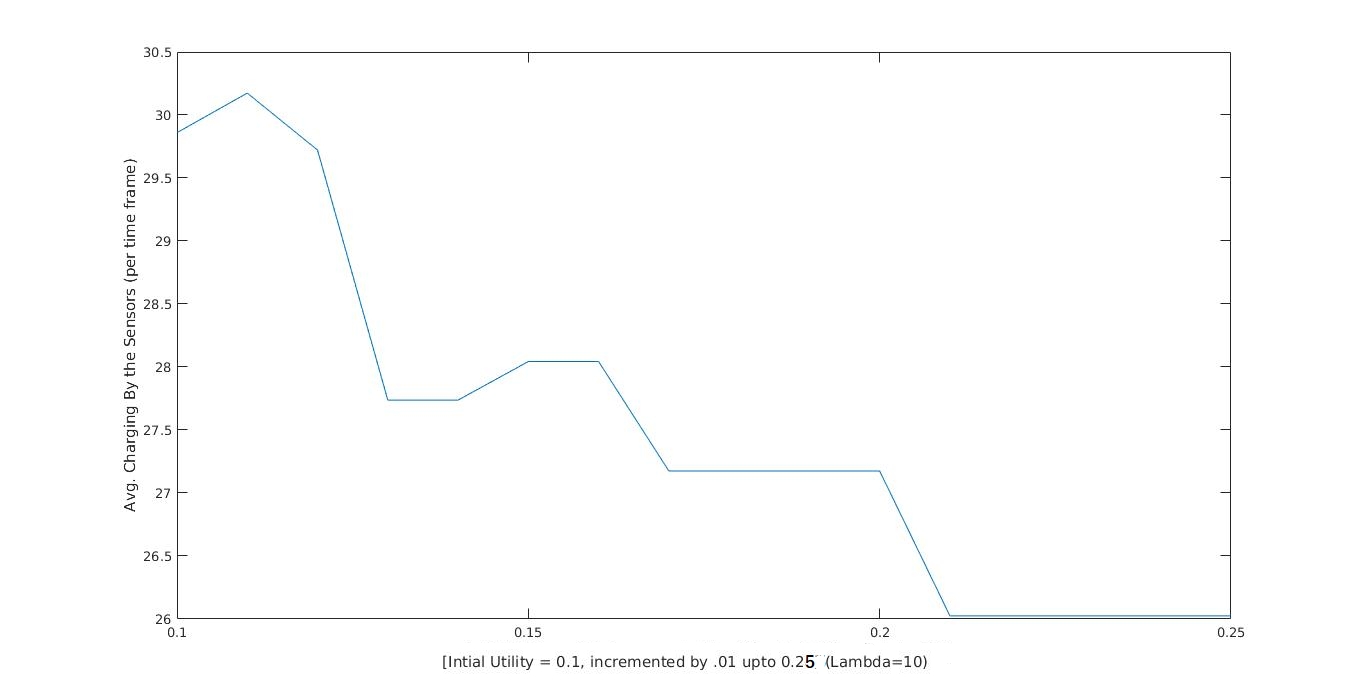} 
       \caption{Average Harvested Charge for TENP with varying $u$ and fixed$\;$ $\lambda$}
   \end{minipage}
   \hfill
   \begin{minipage}{0.8\linewidth}
       \centering
       \label{2ratio}
       \includegraphics[width=1\linewidth]{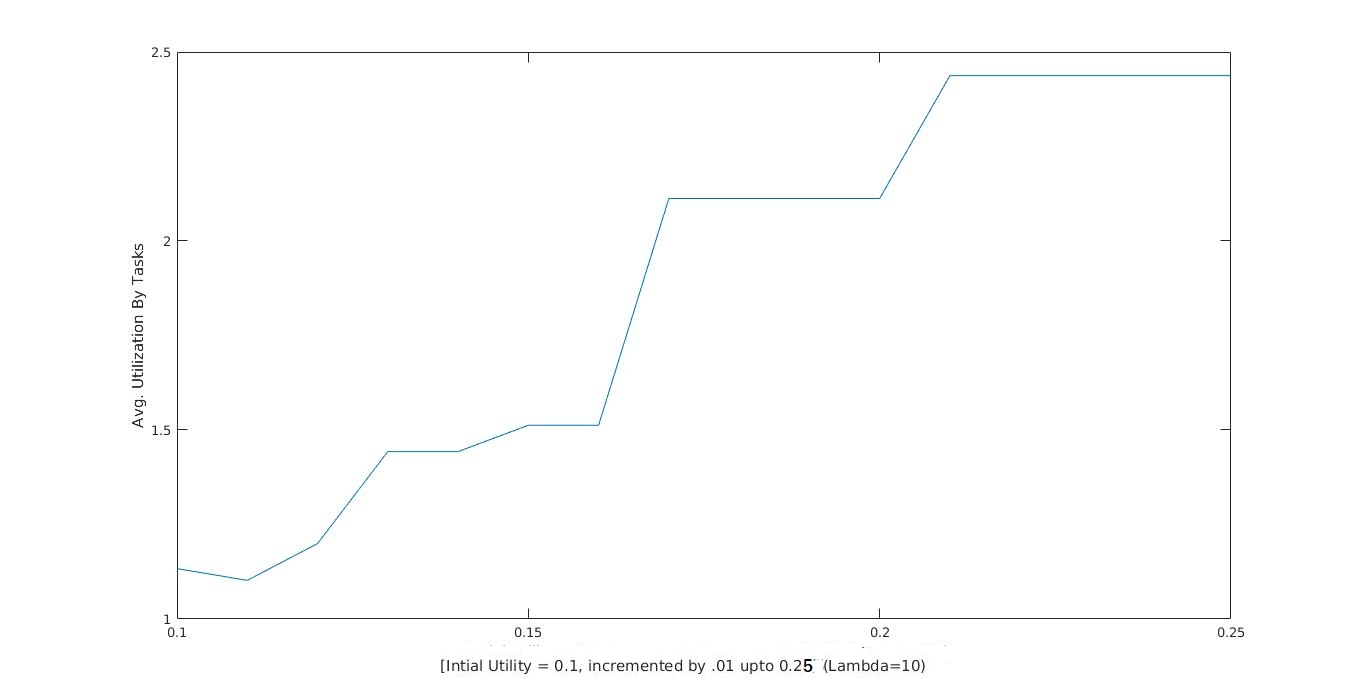} 
       \caption{Average Task Utility for TENP with varying $u$ and fixed $\lambda$}
   \end{minipage}
\end{figure}

\subsubsection{Experiment with ESP and TSP}
As task utility requirement constraint is disregarded in ESP, for this version we have only one setting, where we increase charging requirement $\lambda$. for the ESP problem, Figure 6 and 7 shows how average harvested charge and average task utility changes with the increase of the values of $\lambda$.
\begin{itemize}
\item With the increase of $\lambda$, average harvested charge never decreases (Figure 6).
\item However, the average task utility metric for ESP does not exhibit any clear trend (Figure 7). With relatively smaller values of $\lambda$ (for example values from from 7 to 18), first the average task utility decreases as expected, then increases again for larger values of $\lambda$ (for example values from 21 to 33). One possible explanation of this phenomenon can be this:  For larger values of $\lambda$, the sensor nodes are placed closer to ETs. \textit{Incidentally}, these fixed sensor locations are also closer to locations of tasks of these associated sensor. As a result average task utility increases. We \textit{conjecture} that this particular phenomenon is a result of the \textit{specific tasks and ET placement} for the environment setting shown in Table 2.
\end{itemize}
\begin{figure}[H]
   \centering
   \begin{minipage}{0.8\linewidth}
       \centering
               \label{3CPUTime}
       \includegraphics[width=1\linewidth]{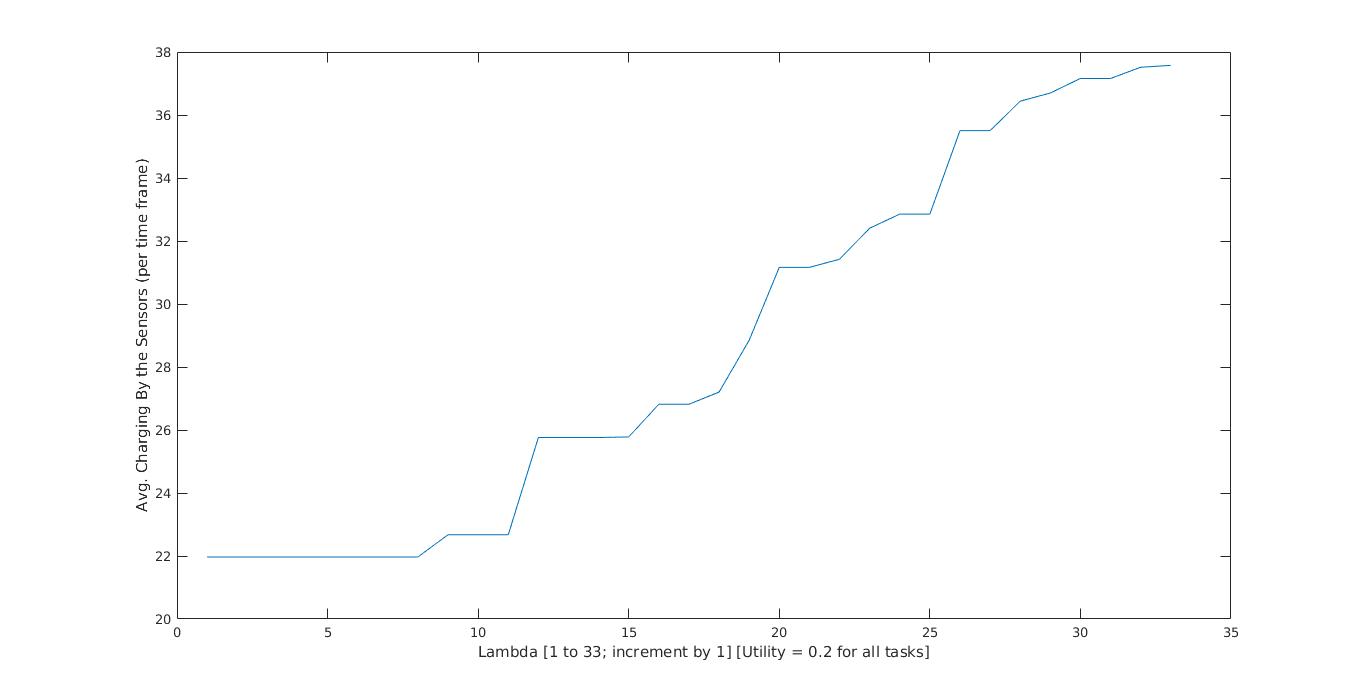} 
       \caption{Average Harvested Charge for ESP with varying $\lambda$}
   \end{minipage}
   \hfill
   \begin{minipage}{0.8\linewidth}
       \label{3ratio}
       \centering
       \includegraphics[width=1\linewidth]{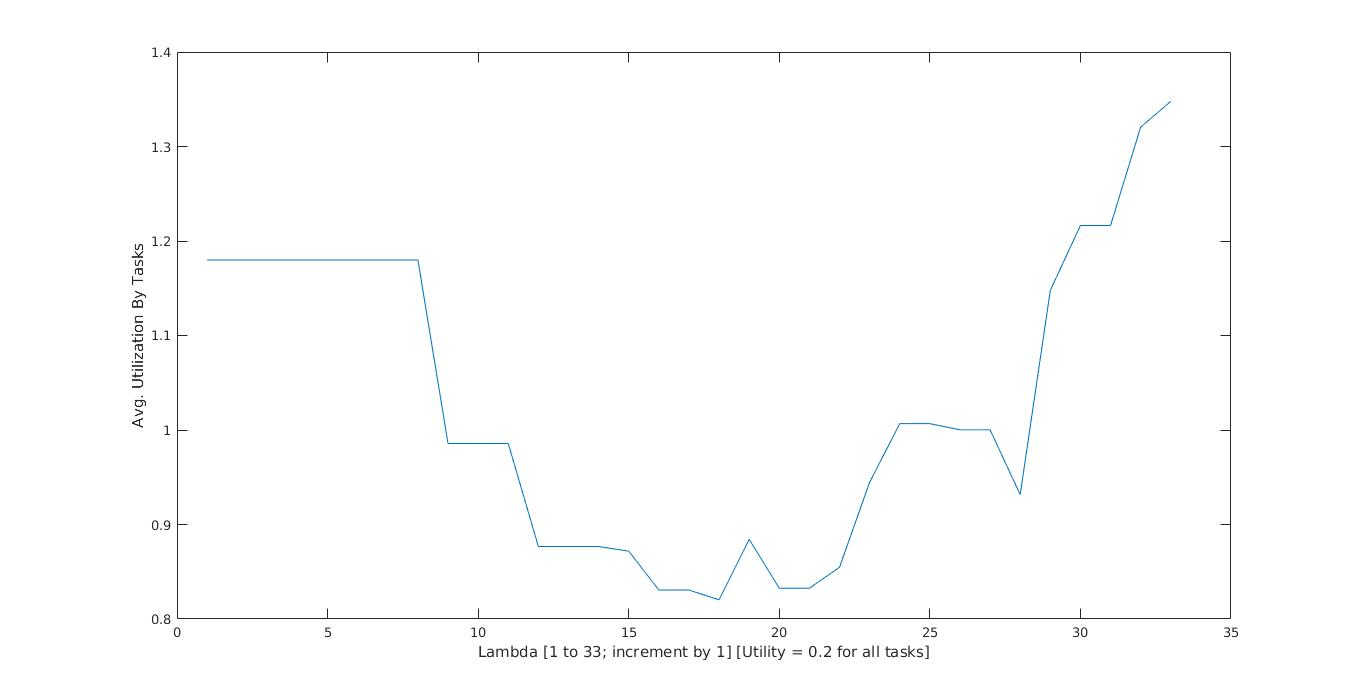} 
       \caption{Average Task Utility for ESP with varying $\lambda$}
   \end{minipage}
\end{figure}
For TSP, as we disregard the charging requirement constraint ($\lambda$), we consider only one experimental setting, where we increase task utility $u$ from $(0.1,0.1,0.1,0.1)$ to $(0.24,0.24,0.24,0.24)$. Figure 8 and 9 shows average harvested charging and average task utility  with varying $u$ values.
\begin{figure} [H]
   \centering
   \begin{minipage}{0.8\linewidth}
       \centering
               \label{3CPUTime}
       \includegraphics[width=1\linewidth]{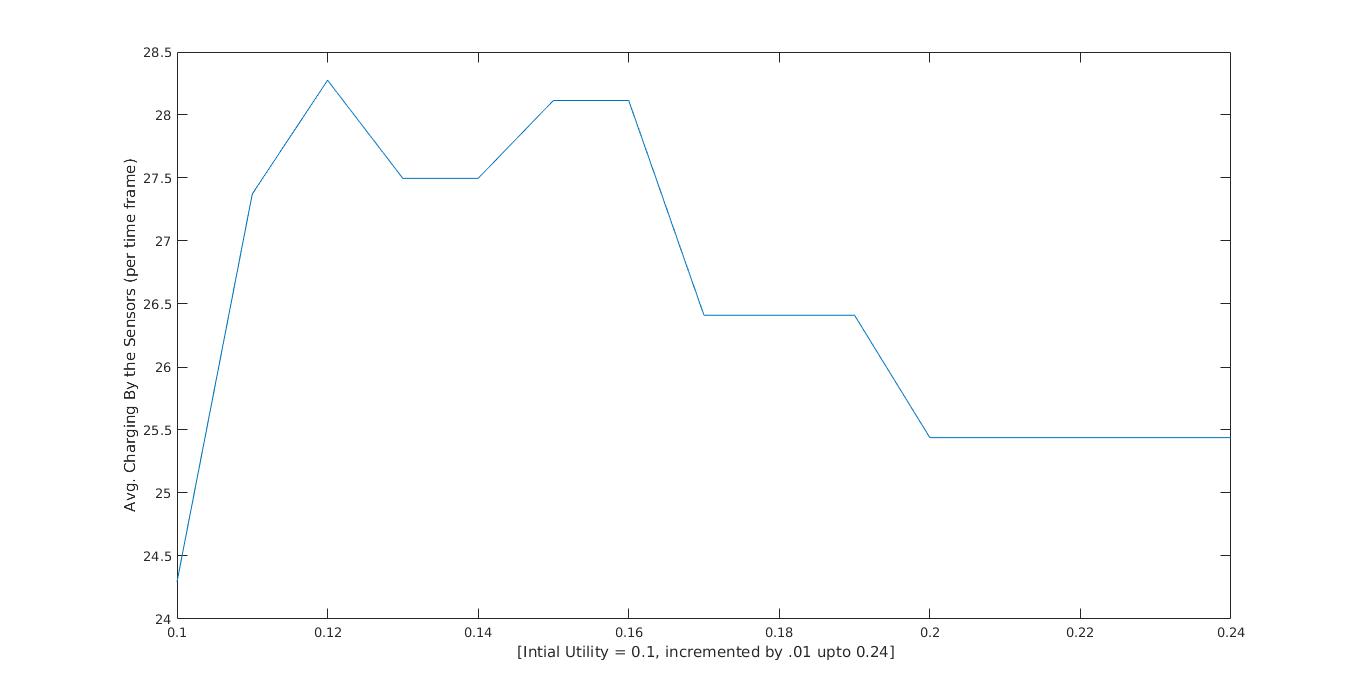} 
       \caption{Average Harvested Charge for TSP with varying $u$}
   \end{minipage}
   \hfill
   \begin{minipage}{0.8\linewidth}
       \label{ratio}
       \centering
       \includegraphics[width=1\linewidth]{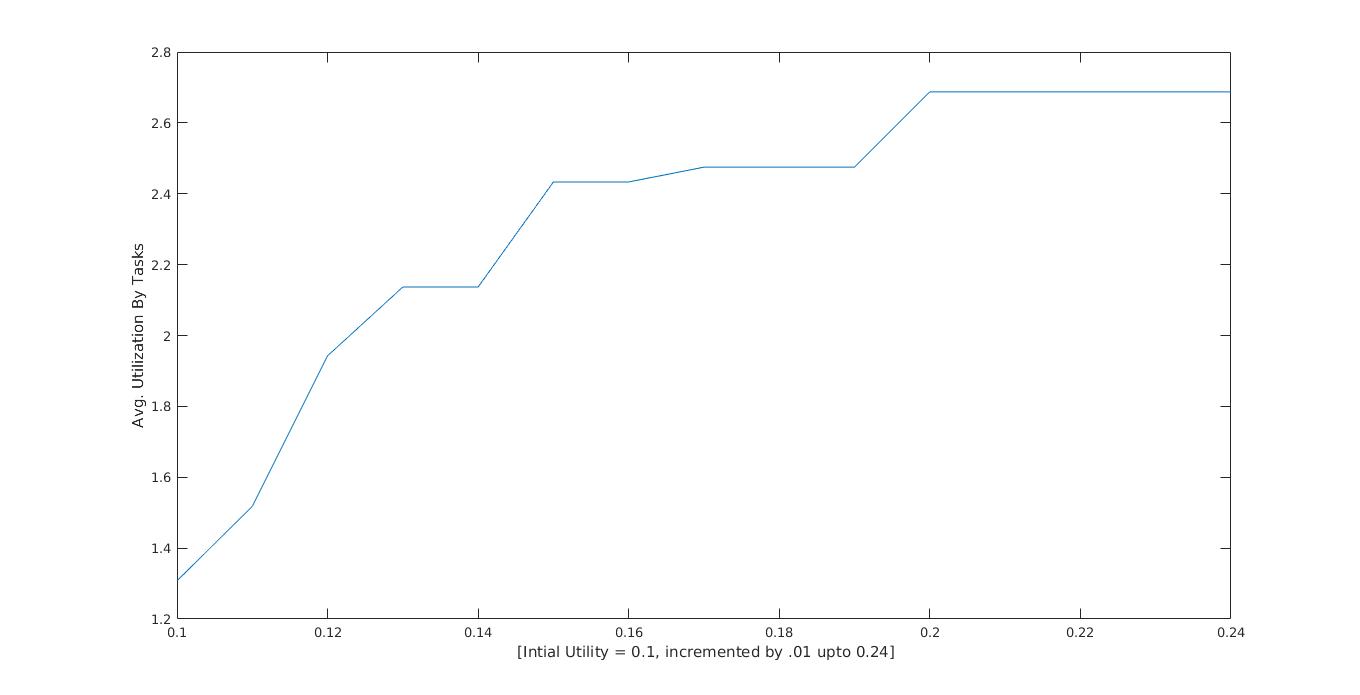} 
       \caption{Average Task Utility for TSP with varying $u$}
   \end{minipage}
\end{figure}
\begin{itemize}
\item In Figure 8, for TSP, we observe that with the increase of task utility requirement ($u$), average harvested charge generally decreases, with a sharp increase for lower values of $u$ (0.1 to 0.12). This sharp increase in average harvested charge is counter intuitive. One possible explanation of this phenomenon is this: when task utility requirements ($u$) are low, sensor nodes can be placed far apart from the tasks. This let sensor nodes to be placed near the ETs and as a result the average harvested charge  increases. But, as values in $u$ increases, the sensor nodes needs to be placed nearer to the tasks and move far apart from the ETs. As a result, average harvested charging decreases, as $u$ values increases.
\item For TSP,  average task utility never decreases  with the increase of utility requirement of tasks (Figure 9). This is expected, as we increase the utility demand for tasks, the nodes are placed nearer to the tasks and as a result average task utility increases.
\end{itemize}

\begin{table}[]
\centering
\caption{Comparison of TENP, ESP and TSP}
\label{my-label}
\begin{tabular}{|l|l|l|}
\hline
\textbf{Version} & \textbf{\begin{tabular}[c]{@{}l@{}}Maximum Average \\ Task Utility\end{tabular}} & \textbf{\begin{tabular}[c]{@{}l@{}}Maximum Average \\ Charging Harvest (Watts)\end{tabular}} \\ \hline
TENP             & 2.43 (in Setting II)                                                                & 29.57  (in Setting I)                                                                           \\ \hline
ESP              & 1.34                                                                             & \textbf{38.57}                                                                               \\ \hline
TSP              & \textbf{2.68}                                                                    & 25.43                                                                                        \\ \hline
\end{tabular}
\end{table}

Table 3 shows the \textit{maximum average harvested charge} and \textit{maximum average task utility} for the three versions of TENP. While simulation with ESP yields the \textbf{highest} \textit{maximum average harvested charge} (38.57 Watts), \textbf{highest} \textit{maximum average task utility} (2.68) is achieved by TSP. Note that with task utility requirement being disregarded in ESP, it has achieved \textbf{lowest} \textit{maximum average task utility} (1.34) in comparison to the other two versions. Similarly, for TSP, with charging requirement being disregarded, it yields \textbf{lowest} \textit{maximum average harvested charge} (25.43 Watts) in comparison to the other two versions. TENP, on the other hand, exhibits a good \textit{balance} between average harvested charge (29.57 Watts) and average task utility (2.43). The results shown in Table 3 are really intuitive.
\section{Future Work}

In the future, we plan to devise a \textit{complete method} for TENP. Once it fails to assign a position for a sensor node, the complete method will need to consider reassigning previously assigned positions via backtracking. This warrants a major revision of Algorithm 1. The complete algorithm is most likely to have an exponential time complexity. Thus, designing good heuristics for the complete algorithm will be another interesting future direction for this work. Devising Satisfiability (SAT)/Mixed Integer Programming (MIP) encoding for the TENP problem seems to be another interesting and natural direction. Specially, with the realization of SAT encoding of TENP, the NP-completeness of TENP will be proven. Additionally, it will be interesting to see how the complete algorithm will perform in comparison to highly efficient modern SAT/MIP solvers on the TENP problem.

\section*{Acknowledgement}
This work was done for a graduate course project at the Computing Science department of the University of Alberta. We thank the instructor Dr. Janelle Harms for her feed-back and suggestions during this course project.

\end{document}